\documentclass[letterpaper, 10 pt, conference]{ieeeconf}  % Comment this line out if you need a4paper

\IEEEoverridecommandlockouts                              % This command is only needed if 
\usepackage{times,xcolor}
\usepackage{amsmath,amssymb,amsfonts,booktabs}
\usepackage{algorithm}
\usepackage{algpseudocode}
\usepackage{bm}
\usepackage{mathtools}
\usepackage{multicol}
\usepackage{graphicx}
\usepackage{placeins}
\usepackage{cite}
\usepackage{censor}
\usepackage[hidelinks]{hyperref}
\usepackage{float}
\usepackage{tabularx}
\usepackage{makecell}
\usepackage{eso-pic}
\newtheorem{theorem}{Theorem}
\newtheorem{proposition}[theorem]{Proposition}
\newtheorem{lemma}[theorem]{Lemma}
\newtheorem{remark}{Remark}
\title{\LARGE \bf
Statistical Learning of Contractive Dynamical Representations\\for Composite Adaptive Control}

\author{\censor{Min Kim$^{1, *}$, Jos\'e Leonardo Brenes$^{1, *}$, Fred Hadaegh$^{1,2}$, Soon-Jo Chung$^{1,2}$}%
\thanks{\censor{* These authors contributed equally.}}%
\thanks{\censor{$^{1}$ California Institute of Technology (Caltech), Pasadena, CA 91125 USA.}}%
\thanks{\censor{$^{2}$ Jet Propulsion Laboratory (JPL), Pasadena, CA 91109 USA.}}%
\thanks{{\censor{\tt\small \{mink, jbrenes, hadaegh, sjchung\}@caltech.edu}}}
}

\begin{document}
\StopCensoring

\AddToShipoutPictureFG*{%
\AtPageUpperLeft{%
\put(\LenToUnit{0.5in},\LenToUnit{-0.25in}){%
      \makebox(0,0)[lt]{%
        \parbox[t]{\dimexpr\paperwidth-1in\relax}{%
          \normalfont
          \fontsize{6.5}{7.5}\selectfont
          \color{black}
          \raggedright
          \copyright~2026 IEEE.  Personal use of this material is permitted.  Permission from IEEE must be obtained for all other uses, in any current or future media, including reprinting/republishing this material for advertising or promotional purposes, creating new collective works, for resale or redistribution to servers or lists, or reuse of any copyrighted component of this work in other works.\\
          Accepted to the 2026 IEEE/RSJ International Conference on Intelligent Robots and Systems (IROS 2026).
          \par
        }%
      }%
    }%
  }%
}

\maketitle
\thispagestyle{empty}
\pagestyle{empty}

%%%%%%%%%%%%%%%%%%%%%%%%%%%%%%%%%%%%%%%%%%%%%%%%%%%%%%%%%%%%%%%%%%%%%%%%%%%%%%%%
\begin{abstract}
We present a representation-learning framework for composite adaptive tracking control under dynamically coupled disturbances.  The framework connects classical disturbance-accommodating control (DAC) to recent last-layer adaptive disturbance-rejection methods. Specifically, we introduce a statistically principled hard expectation--maximization (hard-EM) procedure, with a Kalman smoother in the hard E-step, to identify dynamical representations of disturbance whose latent evolution is uniformly contractive. The learned representation evolves a latent disturbance-excitation state from measured plant features and control inputs and decodes that state into the time-varying disturbance acting on the nominal plant, thereby extending prior ``fixed-decay'' last-layer adaptive methods to a learned, predictive DAC-style formulation. Combined with Bayesian filtering of the learned latent state, this representation yields a composite adaptive tracking controller with predictive capability and provable exponential convergence to a bounded neighborhood. We validate our approach experimentally on a slippery ground vehicle carrying a liquid-sloshing tank and a pendulum load, and we further assess its robustness on a system of coupled Duffing oscillators. Across both settings, the method achieves accurate disturbance prediction and improved overall tracking performance relative to fixed-decay representation-learning ablations, LTI disturbance-accommodating baselines, and model-based PD baselines.
\end{abstract}

%%%%%%%%%%%%%%%%%%%%%%%%%%%%%%%%%%%%%%%%%%%%%%%%%%%%%%%%%%%%%%%%%%%%%%%%%%%%%%%%
\section{INTRODUCTION}
There is a substantial body of literature on \emph{representation learning} for \emph{last-layer adaptive control} and disturbance rejection \cite{OCSh2022, LuXi2025, YiOn2025}; however, many existing applications focus on \emph{external} disturbances such as aerodynamic effects or terrain-induced slip. In such settings, disturbances are captured by a latent variable that is treated as \emph{exogenous}; its value shifts as the ``environment'' changes, sometimes with an assumed decay rate, which corresponds to the Ornstein--Uhlenbeck process. In contrast, when the unknown disturbance is \emph{dynamically coupled} with the controlled plant, such as liquid sloshing in a vehicle's tank, it is natural to view the latent variable as an augmented dynamical state that evolves according to its own dynamics, in the spirit of disturbance-accommodating control (DAC). In such cases, we expect adaptive controllers with predictive capability to have a comparative advantage over counterparts that do not incorporate prediction. This motivates the learning of dynamical representations tailored to adaptive control; in this paper, we develop a representation that captures the coupled disturbance source's excitation mechanism and its influence on the main system.

\begin{figure} [!t]
    \centering
    \includegraphics[width=\linewidth]{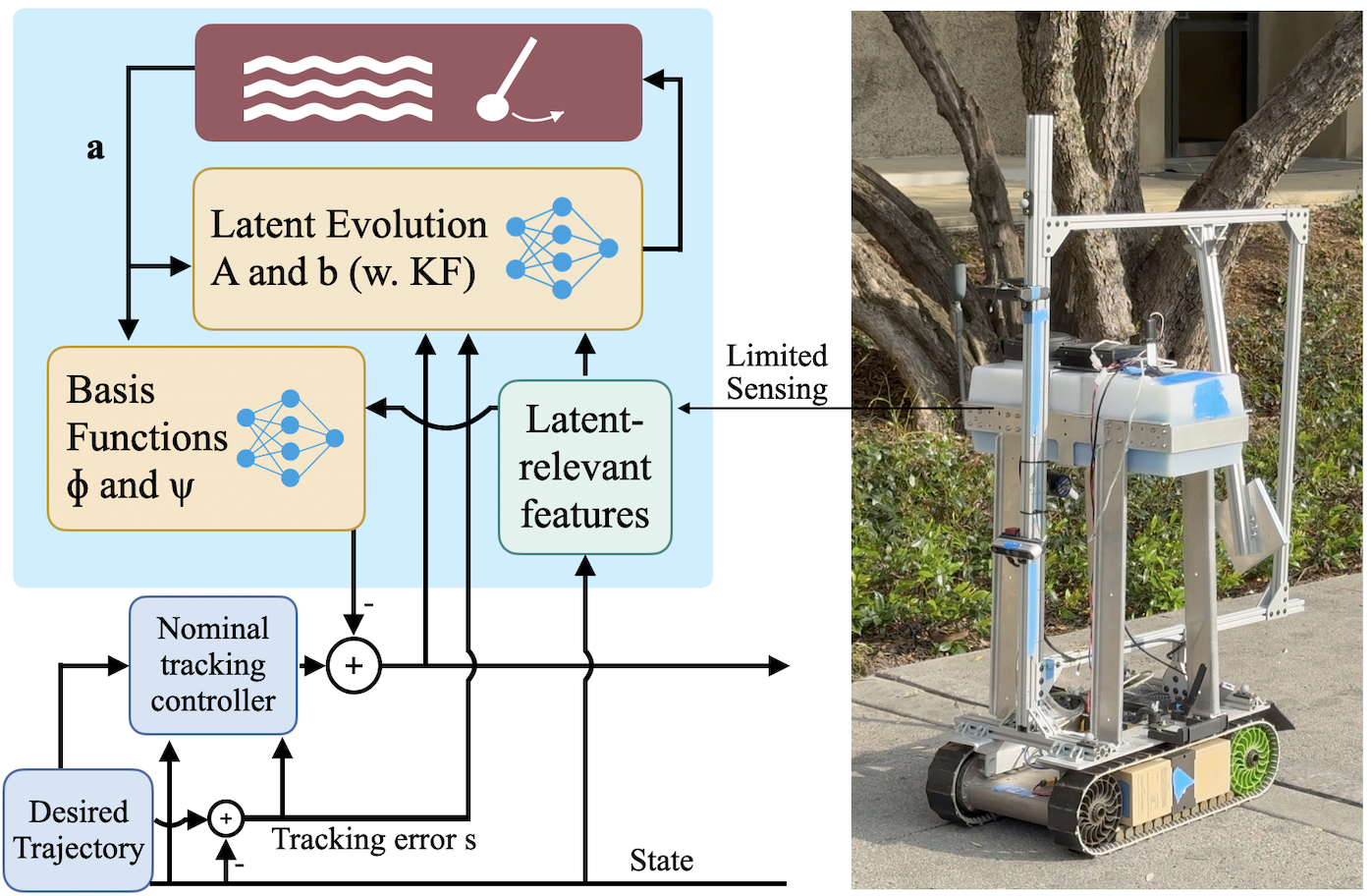}
\caption{Proposed adaptive control architecture and experimental platform. Left: block diagram of the controller. A latent state is estimated online via a Kalman--Bucy-style filter from the disturbance proxy $y$, using a latent-state evolution model parameterized by $(A,b)$. The basis functions $(\Phi,\psi)$ map latent-relevant features $\phi$ to an adaptive disturbance-rejection term that augments a nominal tracking controller. The quantities $(A,b,\Phi,\psi)$ are parameterized by neural networks and learned using a statistically principled training procedure. Right: tracked mobile robot used for experimental validation.}
    \label{fig:figureone}
    \vspace{-10pt}
\end{figure}

The contributions of this paper are threefold: \emph{(i)} we introduce a statistically principled hard expectation--maximization training framework, with Kalman-smoother-based latent inference (Fig.~\ref{fig:figuretwo}), for learning \emph{contractive dynamical latent representations} of disturbance from data; \emph{(ii)} we design a neural architecture for the dynamical representation that enforces uniform contraction, is compatible with Bayesian filtering, and yields predictive disturbance estimates with stability guarantees when used in a composite adaptive controller; and \emph{(iii)} we demonstrate the method on hardware and in simulation, showing accurate disturbance prediction and improved tracking over fixed-decay representation-learning ablations and LTI-based DAC. Taken together, these results show that statistically principled learning of contractive latent disturbance dynamics can bridge representation learning and composite adaptive control for systems with dynamically coupled disturbances.
\begin{figure} [!h]
\vspace{5pt}
    \centering
    \includegraphics[width=\linewidth]{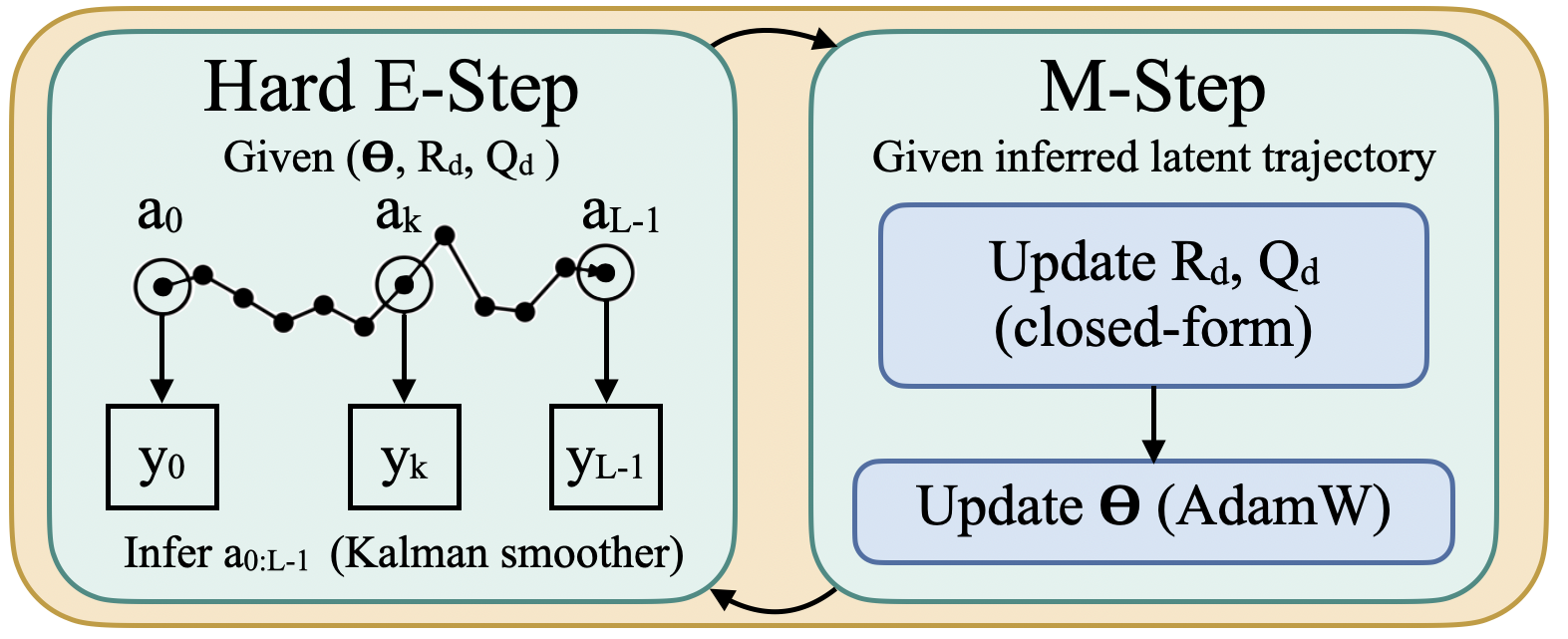}
    \vspace{-15pt}
    \caption{Statistically principled training via hard-EM with exact latent inference. The procedure alternates between a hard E-step and an M-step: given current parameters $(\theta,R_d,Q_d)$, the hard E-step infers the latent trajectory $\mathbf a_{0:L-1}$ for each data window, and the M-step updates the noise covariances $(R_d,Q_d)$ from residual statistics and updates the network parameters $\theta$.}
    \label{fig:figuretwo}
    \vspace{-15pt}
\end{figure}

We next position our paper relative to prior work. From the perspective of \emph{classical control}, our approach is closely related to disturbance-accommodating control (DAC), which represents disturbance using an augmented disturbance state with its own dynamics and estimates this augmented state online for compensation \cite{Jo1986}. DAC motivates treating the disturbance as an estimable and predictable dynamical quantity rather than as a memoryless exogenous input.
Our perspective is also connected to the robust servomechanism problem \cite{Da1976}, where disturbance signals are generated by an exosystem and robust output regulation requires embedding an internal model \cite{FrWo1976} of the disturbance in the closed loop. From this standpoint, our method can be interpreted as learning a data-driven, contractive analogue of the disturbance generator in the spirit of identification-based internal-model regulation \cite{BiMa2019}. In contrast to \cite{BiMa2019}, which adapts an internal model by linear regression, we identify an affine-in-latent model via statistical learning and exploit Kalman-based estimation.

\emph{Learning-based control} has been demonstrated on physical robot platforms, including quadrotors under wind disturbance \cite{OCSh2022, pmlr-v155-joshi21a}, drone racing \cite{SoSt2021}, learning-based MPC for quadrotors \cite{ChSi2023} and autonomous car racing \cite{KaHe2019}, rapid motor adaptation for bipedal locomotion \cite{KuLi2022}, and robust perceptive locomotion for quadrupedal robots in natural environments \cite{MiLe2022}. These hardware demonstrations highlight the importance of controllers that remain reliable under modeling errors and disturbances, which are often induced or amplified by robotic maneuvers. Motivated by this need, recent learning-based MPC approaches improve tracking under model mismatch by incorporating learned disturbance or residuals into the MPC framework \cite{KrZa2025,ZhGe2025}. Also, \cite{10288520} learns a Koopman-based disturbance dynamics model and uses the corresponding disturbance observer for compensation. In contrast, we learn a latent disturbance dynamics model by a Kalman-smoother-based latent trajectory inference and hard expectation--maximization (EM). We note that the learning strategy presented in this paper is analogous to previous smoother-based EM methods for system identification \cite{GuCh2024,WaLi2025}, although these papers do not involve neural networks.

In this paper, we focus on disturbance representations that are uniformly contractive \cite{LoSl1998, Bu2026} and are affine in the latent variable (see (\ref{eq:latent_model})). The contraction property of the dynamical representation is required for insensitivity to the \emph{unknown} initial latent state (Proposition~\ref{prop:insensitivity_IC}), and is physically plausible for dissipative disturbance dynamics. On the other hand, the affine-in-latent form enables closed-form Bayesian estimation. This motivates us to use a hard-EM procedure to train the neural networks (Fig.~\ref{fig:figuretwo} and Section~\ref{sec:nn_arch_training}); during training, the latent trajectory is repeatedly estimated by Kalman smoothing. Moreover, using the Kalman--Bucy filter for the latent state, we synthesize a composite adaptive controller~\cite{slotine1989composite,OCSh2022,LuXi2025}, whose exponential tracking guarantees are analyzed in Theorems~\ref{thm:2ndorder} and \ref{thm:1storder}. This analysis relies on a covariance bound for uniformly contractive systems (Lemma~\ref{lemma:bounded_cov}).

In Section~\ref{sec:exp_results}, we evaluate the learned dynamical representation and the resulting adaptive controller on two different scenarios. First, we study a ground vehicle \censor{(``GVR-Bot'')} equipped with a liquid-sloshing tank, a pendulum load, and an unknown, inaccessible internal PID controller; vehicle maneuvers excite both the fluid sloshing and the pendulum motion. Then, we consider a system of coupled Duffing oscillators as a numerical testbed, where the robustness of the proposed tracking controller is assessed under increasingly adverse Gaussian-noise conditions. These examples show not only that the method can learn predictive maneuver-to-disturbance mappings, but also that the learned mappings can be readily and effectively incorporated into composite adaptive controllers.

\section{THE CONTROL METHOD AND GUARANTEES}
The main controller equations are given by the theorems in Section~\ref{subsec:controllers}.

\subsection{System Description and the Dynamical Representation}
We consider systems with additive disturbances of the form:
\begin{equation} \label{eq:second_order}
    M(q) \ddot{q} + C(q, \dot q) \dot q+ g(q) = u + d
\end{equation}
where the coefficient functions are continuously differentiable and $\dot M - 2C$ is skew-symmetric. We also consider the simpler first-order system:
\begin{equation} \label{eq:first_order}
    \dot v = f(v) + u + d,
\end{equation}
with continuously differentiable $f$. We aim to learn disturbance estimators of the form \eqref{eq:latent_model} for use in composite adaptive control \cite{SlLi1991}:
\begin{subequations}\label{eq:latent_model}
\begin{align}
      \dot {\mathbf{a}} &= A(\phi(t), u(t)) \mathbf{a} + b(\phi(t), u(t)), \label{eq:latent_dyn}\\
      \hat d &= \Phi(\phi(t)) \mathbf{a} + \psi (\phi(t)),
\end{align}
\end{subequations}
where $A: \mathbb{R}^{d_\phi + d_u} \to \mathbb{R}^{d_a \times d_a}$, $b: \mathbb{R}^{d_\phi + d_u} \to \mathbb{R}^{d_a}$, $\Phi: \mathbb{R}^{d_\phi} \to \mathbb{R}^{d_u \times d_a}$, $\psi: \mathbb{R}^{d_\phi}  \to \mathbb{R}^{d_u}$ are continuous functions to be learned.
 
Intuitively, $\mathbf{a}\in \mathbb{R}^{d_a}$ summarizes disturbance-relevant hidden dynamics; $A$ and $b$ define its evolution, while $\Phi$ and $\psi$ decode it into a disturbance estimate. We emphasize that $\phi:\mathbb{R}_{\ge 0} \to \mathbb{R}^{d_\phi}$ is a continuous ``feature'' signal which is determined causally. In most cases $\phi$ is simply the state: $\phi=(q, \dot q)$, but it could additionally contain measurements from on-board sensors, e.g., a camera. Here, $\mathbf{a}$ statistically explains the disturbance $d$ and evolves according to the dynamic model \eqref{eq:latent_dyn}. Unless otherwise stated, the control signals $u: \mathbb{R}_{\ge 0} \to \mathbb{R}^{d_u}$ are continuous and $\| \cdot \|$, when applied to a matrix, denotes the spectral norm.

Note that in~\eqref{eq:latent_model} the estimate $\hat d(t)$ does not depend on the instantaneous control input $u(t)$. This restriction may be undesirable in some settings. An alternative is to linearly approximate the map $u(t) \mapsto d(t)$; to do so, set
\begin{equation}\label{eq:affine_readout}
    \hat d_k(t) \;=\; u_{\text{aug}}(t)^\top \bigl(\Phi_k \mathbf{a} + \psi_k \bigr)
\end{equation}
for each disturbance coordinate $1 \le k\le \operatorname{dim}(d)$, where the augmented input $u_{\text{aug}}(t)\coloneqq [u(t)^\top,\,1]^\top$. This parameterization is convenient for controller synthesis: since $u(t)$ enters \eqref{eq:affine_readout} \emph{affinely}, the relation $u(t)+\hat d(t)=u_{\text{desired}}(t)$ can be easily solved for $u(t)$; in practice, one solves the Tikhonov-regularized normal equations centered at $u_{\text{desired}}(t)$ for robustness. Finally, we also note that a prototypical class of examples for~\eqref{eq:latent_model} consists of multi-component mechanical systems coupled through force elements.

Hereafter, we assume access to an online noisy estimate $y(t)=d(t)+\epsilon(t)$, where $\epsilon(t)$ is a continuous, additive error. If $y(t)$ were exact and the control loop were sufficiently fast relative to the disturbance timescale, direct cancellation would largely suffice; in practice, however, $y(t)$ is often obtained by numerically differentiating the measured $q$ or $v$, which amplifies sensor noise, and actuator delays further limit the feasibility of instantaneous cancellation for rapidly varying disturbances.

The proposed learned dynamical representation instead produces systematically refined estimates of the current disturbance $d(t)$ via the Kalman--Bucy filter, by fusing the full measurement history $y(\cdot)$ and, implicitly, all measurements available in the training set. When the disturbance source is dynamically coupled to the plant, learning an explicit latent dynamics is a natural design choice; these learned dynamics then serve as predictive process models in the Kalman--Bucy filter. Moreover, when constructing disturbance labels for training, one may leverage higher-quality off-board sensing and non-causal signal processing to obtain more reliable numerical derivatives than would be available on-board in real time. The refined estimates from the filter are then used for disturbance compensation, thereby yielding composite adaptive tracking controllers with provable convergence guarantees (see Theorems~\ref{thm:2ndorder} and \ref{thm:1storder}).

Before proving the main Theorems \ref{thm:2ndorder} and \ref{thm:1storder} which give the controllers (\eqref{eq:feedback_2ndorder} or \eqref{eq:feedback_1storder} with \eqref{eq:KalmanComposite}), we first formalize the intuition that uniformly contractive linear systems are insensitive to initial conditions which will be \emph{unknown} in practice.
\begin{proposition} [Exponential insensitivity to initial condition] \label{prop:insensitivity_IC}
Suppose $\frac 1 2 \left( A+A^\top \right) \preceq - \lambda I$ and $\|\Phi\| \le C$ pointwise for some constants $\lambda>0$ and $C>0$. For a fixed pair $\phi$ and $u$, any two trajectories $\mathbf{a}_1$ and $\mathbf{a}_2$ of \eqref{eq:latent_model} satisfy
\begin{equation}
   \|\hat d_1 (t) - \hat d_2(t)\| \le C e^{-\lambda t} \|\mathbf{a}_1(0) - \mathbf{a}_2(0)\|
\end{equation}
for all $t \ge 0$.
\end{proposition}
\begin{proof}
We have $\frac{d}{dt} \left(\left\|\mathbf{a}_1(t)-\mathbf{a}_2(t)\right\|^2 \right) = 2\left(\mathbf{a}_1(t)-\mathbf{a}_2(t) \right)^\top \left(\dot {\mathbf{a}}_1(t)-\dot {\mathbf{a}}_2(t) \right) =\left(\mathbf{a}_1(t)-\mathbf{a}_2(t) \right)^\top (A+A^\top)\left(\mathbf{a}_1(t)-\mathbf{a}_2(t) \right) \le -2 \lambda \left\|\mathbf{a}_1(t)-\mathbf{a}_2(t)\right\|^2$. Therefore, $\left\|\mathbf{a}_1(t)-\mathbf{a}_2(t)\right\| \le e^{-\lambda t} \left\|\mathbf{a}_1(0)-\mathbf{a}_2(0)\right\|$. Now, the claimed result follows from $\hat d= \Phi \mathbf{a} + \psi$.
\end{proof}

Moreover, uniformly contractive linear systems lead to bounded covariance estimates when the Kalman--Bucy filter is used. Specifically, we consider the following covariance evolution equation:
\begin{equation} \label{eq:covariance}
    \dot P = A(t) P + P A (t)^\top + Q - P \Phi(t)^\top R^{-1} \Phi(t) P
\end{equation}
with a fixed initial condition $P(0)=P_0 \succ0$ and $Q \succeq 0$, $R \succ 0$. Here, $A(t)\coloneqq A(\phi(t),u(t))$ and $\Phi(t) \coloneqq \Phi(\phi(t))$.
\begin{lemma} \label{lemma:bounded_cov}
Suppose $\frac 1 2 \left( A+A^\top \right) \preceq - \lambda I$ for some constant $\lambda>0$. Then the unique solution of \eqref{eq:covariance} is positive definite and satisfies
\begin{equation}
\lambda_{\max}\bigl(P(t)\bigr) \le e^{-2\lambda t}\lambda_{\max}(P_0) + \frac{\lambda_{\max}(Q)}{2\lambda} \bigl(1-e^{-2\lambda t}\bigr)
\end{equation} for all $t \ge 0$.
Moreover, if $\Phi$ is bounded, $Q \succ 0$, and $\frac{1}{2} (A+A^\top)\succeq -\mu I$, then $\lambda_{\min}\left(P(t)\right) \ge p_*$ for some positive $p_*$ given by an explicit formula.
\end{lemma}
\begin{proof}
The positive definiteness claim follows from the comparison theorem for Hermitian RDE~\cite[Theorem~4.1.4]{AbFr2003} (compare to the same equation but with $Q=0$), and existence of $P$ on $[0,\infty)$ follows from Theorem~4.1.6.
    
For the upper bound, consider the linear equation $\dot {\tilde P} = A \tilde P + \tilde P A^\top + Q$ with the initial condition $\tilde P(0)=P_0$. Theorem~4.1.4 gives $P(t) \preceq \tilde P(t)$ for all $t \ge 0$. Since $\tilde P(t) = \Phi_A(t,0)P_0\Phi_A(t,0)^\top + \int_0^t \Phi_A (t,\tau) Q \Phi_A (t,\tau)^\top d\tau$ for all $ t \ge 0$, and $\|\Phi_A(t,t_0)\|\le e^{-\lambda (t- t_0)}$ for all $t \ge t_0 \ge 0 $, we get $\lambda_{\max}(\tilde P(t)) = \|\tilde P (t)\| \le e^{-2\lambda t}\lambda_{\max}(P_0) + \frac{\lambda_{\max}(Q)}{2\lambda}(1-e^{-2\lambda  t})$. The claim follows from the Weyl monotonicity principle.

For the lower bound, we have $P(t) \succeq p_* I$ with $p_* = \frac{-\mu + \sqrt{\mu^2 + s_u \lambda_{\min}(Q)} } {s_u} \wedge \lambda_{\min} (P_0)$, where $s_u \coloneqq \frac{\overline\Phi ^2}{\lambda_{\min} (R)}$ and $\overline\Phi>0$ is an upper bound for $\Phi$; this can be proved by comparing the original equation to a related one: $\dot P_c = A P_c + P_c A^\top  + (s_u p_*^2 I  - p_*A - p_* A^\top) - P_c (s_u I) P_c$, $P_c(0) = p_* I$ and noting that $Q \succeq s_u p_*^2 I - p_*A - p_*A^\top$, $\Phi^\top R^{-1}\Phi \preceq s_u I$.
\end{proof}

\subsection{Composite Adaptive Tracking Controller and Analysis} 
\label{subsec:controllers}
While Lemma \ref{lemma:bounded_cov} is of independent interest, its main role in this paper is to establish exponential tracking results for various composite adaptive control strategies based on the Kalman--Bucy filter.

First, we consider a tracking problem for~\eqref{eq:second_order} with a twice continuously differentiable reference trajectory $q_d$. We use Kalman--Bucy-inspired composite adaptation equations of the form:
\begin{subequations} \label{eq:KalmanComposite}
    \begin{align}
        \begin{split}
            \dot {\hat {\mathbf{a}}} &= A(\phi,u) \hat {\mathbf{a}}+ b(\phi,u) \\&+ P \Phi (\phi)^\top R^{-1} \left(y - \Phi(\phi) \hat{\mathbf{a}}-\psi(\phi)\right )+P \Phi(\phi) ^\top s,
        \end{split}\\
        &\dot P = AP + PA^\top + Q -P\Phi^\top R^{-1} \Phi P
    \end{align}
\end{subequations}
where $s \coloneqq \dot{\tilde q} + \Lambda \tilde q$, $\tilde q\coloneqq q-q_d$, and $P(0)=P_0$, $Q$, $R$, $\Lambda \succ 0$. For the main theorems, we assume $d(t) = \Phi(\phi(t)) \mathbf{a}(t) +  \psi(\phi(t)) + r(t)$, so that the disturbance $d$ is realized by continuous signals $\mathbf{a}(t)$ and $r(t)$ with $\mathbf{a}$ being differentiable. Note that \eqref{eq:KalmanComposite} are the standard Kalman--Bucy equations but with $P\Phi^\top s$, an additional composite adaptation term. Note that when the Kalman covariance $P$ becomes larger, the tracking error term $P \Phi(\phi)^\top s$ also becomes larger (along with the innovation term).

\begin{remark}
In~\eqref{eq:KalmanComposite}, $A=A(\phi,u)$, $b=b(\phi,u)$, $\Phi=\Phi(\phi)$, $\psi=\psi(\phi)$ are dependent on the measured features and/or control, and therefore are implicitly time-varying. Previous last-layer adaptive disturbance rejection methods \cite{OCSh2022, LuXi2025} may be interpreted as the case of $A=-\lambda I$, $b = 0$, $\psi = 0$ with hyperparameter $\lambda>0$, although their learning procedures differ from the procedure used here.
\end{remark}

\begin{remark}
\label{remark:data_driven_covariances}
The Kalman-smoother approach to representation learning (Section~\ref{subsec:process_noise_kalman}) provides data-driven estimates of $Q$ and $R$ that serve as principled starting points for subsequent gain tuning.
\end{remark}

Now we present the main theorems demonstrating exponential tracking convergence via the adaptation dynamics \eqref{eq:KalmanComposite}; here, \emph{exponential convergence to a bounded ball} means that there exist constants $c_i>0$ such that each solution on $\mathbb{R}_{\ge 0}$ of the closed-loop system satisfies $\| q(t)- q_d(t)\| \le c_1 e^{-c_2 t}+ c_3$ (or $\|v(t)-v_d(t)\| \le c_1 e^{-c_2 t}+ c_3$ for \eqref{eq:first_order}), where $c_2$ and $c_3$ are independent of $\hat{\mathbf{a}}(0)$, $\mathbf{a}(0)$, and $q(0)$, $\dot q (0)$ (or $v(0)$ for \eqref{eq:first_order}). As expected, the theorems will depend on the bounds on $A \mathbf{a} + b - \dot{\mathbf{a}}$ and $r$, i.e., the quality of the dynamical representation \eqref{eq:latent_model}.

\begin{theorem}[Exp. Tracking: 2nd-order mechanical systems] \label{thm:2ndorder}
Consider the system \eqref{eq:second_order}, \eqref{eq:KalmanComposite} with the feedback controller
\begin{equation} \label{eq:feedback_2ndorder}
    u = M \dot{v}_r + C v_r + g - Ks - \Phi \hat {\mathbf{a}} - \psi,
\end{equation}
where $K \succ 0$ and $v_r \coloneqq \dot q_d - \Lambda \tilde q$. If $-\mu I \preceq \frac 1 2 \left( A+A^\top \right) \preceq - \lambda I$, $\underline{m} I \preceq M \preceq \overline{m} I$ for some positive constants $\mu$, $\lambda$, $\underline{m}$, $\overline{m}$, and if $\left(A\mathbf{a}+b - \dot{\mathbf{a}} \right)$, $\Phi$, $r$, $\epsilon$ are bounded by some constants independent of $\hat{\mathbf{a}}(0)$, $\mathbf{a}(0)$, $q(0)$, $\dot q (0)$, then the feedback controller \eqref{eq:feedback_2ndorder} achieves exponential convergence to a bounded ball.
\end{theorem}
\begin{proof}
First, note that $\dot s = \ddot{q} - \dot v_r$ and therefore $M \dot s = -(K+C) s - \Phi \tilde {\mathbf{a}} + r$, where $\tilde{\mathbf{a}}\coloneqq \hat{\mathbf{a}}-\mathbf{a}$. Moreover, we have $P^{-1}\dot{\tilde{\mathbf{a}}}= \Phi^\top s + (P^{-1} A - \Phi^\top R^{-1} \Phi) \tilde {\mathbf{a}} + P^{-1}(A\mathbf{a} + b - \dot {\mathbf{a}}) + \Phi^\top R^{-1} (r+\epsilon)$.
Defining $V \coloneqq s^\top M s+ \tilde{\mathbf{a}}^\top P^{-1} \tilde{\mathbf{a}}$, we have $\dot V =-2s^\top K s -\tilde {\mathbf{a}}^\top \left( P^{-1}QP^{-1} + \Phi^\top R^{-1}\Phi \right)\tilde{\mathbf{a}} + 2 s^\top r + 2 \tilde {\mathbf{a}}^\top P^{-1}(A\mathbf{a} + b - \dot{\mathbf{a}}) + 2\tilde{\mathbf{a}}^\top\Phi^\top R^{-1} (r+\epsilon)$.

By Lemma~\ref{lemma:bounded_cov}, for a sufficiently small $\alpha>0$, we have $K \succeq \alpha M$ and $Q \succeq 2 \alpha P$ for all $t \ge 0$, which implies $ -2K \preceq - 2 \alpha M$ and $- P^{-1}Q P^{-1}- \Phi^\top R^{-1}\Phi \preceq - 2\alpha P^{-1}$. Therefore, we conclude that $\dot V \le -2 \alpha V + 2 s^\top r + 2 \tilde {\mathbf{a}}^\top P^{-1}(A\mathbf{a} + b - \dot{\mathbf{a}}) + 2\tilde{\mathbf{a}}^\top \Phi^\top R^{-1} (r+\epsilon)$. 
        
Applying Cauchy--Schwarz gives
\begin{equation} \label{eq:Vdot_ub}
    \begin{split}
        \dot V \le &- 2 \alpha V + 2 \sqrt{\frac{V}{\lambda_{\max}(P)^{-1} \wedge \underline{m}}} \\
        &\times \left\| \begin{bmatrix} r \\P^{-1}(A\mathbf{a} + b - \dot {\mathbf{a}} ) + \Phi^\top R^{-1} (r+\epsilon) \end{bmatrix}\right\|.
    \end{split}
\end{equation}
Lemma~\ref{lemma:bounded_cov} with the assumptions of this theorem shows that the norm in \eqref{eq:Vdot_ub} is bounded by a constant, say $\overline{d}\ge 0$. For each small $c>0$, define $W_c = \sqrt{V+c} $. We have $\dot W_c \le -\alpha W_c + \alpha \sqrt{c} + \frac{\overline{d}}{c^*}$ where $c^* \coloneqq \left(\underline{m} \wedge \left( \|P_0\|+\|Q\|/(2\lambda) \right)^{-1}  \right)^{1/2}$. 
        
\raisebox{-2pt}{Differentiating $G_c(t)\coloneqq e^{\alpha t} \left(W_c - \sqrt c - \overline{d}/(c^* \alpha)\right)$ and} passing to the limit $c \to 0+$ gives:
\begin{equation}
    \sqrt{\underline{m}} \|s(t)\| \le \sqrt{V (t)} \le e^{-\alpha t} \left( \sqrt{V(0)} - \frac{\overline{d}}{c^* \alpha}\right) +\frac{\overline{d}}{c^* \alpha}  .
\end{equation}

By solving the linear system $\dot{\tilde{q}} = -\Lambda \tilde{q} + s$ in $\tilde q$, using $\|\exp(-t\Lambda ) \|= \exp\left(-t\lambda_{\min} (\Lambda)\right)$ for all $ t \ge 0$, and choosing $\alpha <\lambda_{\min}(\Lambda)$ if necessary, we conclude that $\| \tilde q\| \le c_1 e^{-c_2 t}+ c_3$ for some $c_i >0$ with $c_2$, $c_3$ being independent of $\hat{\mathbf{a}}(0)$, $\mathbf{a}(0)$, $q(0)$, and $\dot q (0)$.
\end{proof}

We also present a simpler analogue of the previous theorem in the context of \eqref{eq:first_order}.
\begin{theorem}[Exp. Tracking: 1st-order systems]
\label{thm:1storder}
Consider the system \eqref{eq:first_order}, \eqref{eq:KalmanComposite} with the feedback controller
\begin{equation}
\label{eq:feedback_1storder}
    u = \dot v_d - f(v)  -Ks - \Phi \hat {\mathbf{a}} - \psi ,
\end{equation}
where $K \succ 0$ and now $s \coloneqq v- v_d$ with a continuously differentiable reference velocity $v_d$. If $-\mu I \preceq \frac 1 2 \left( A+A^\top \right) \preceq - \lambda I$ for some positive constants $\mu$, $\lambda$, and if $\left(A\mathbf{a}+b - \dot{\mathbf{a}} \right)$, $\Phi$, $r$, $\epsilon$ are bounded by some constants independent of $\hat{\mathbf{a}}(0)$, $\mathbf{a}(0)$, $v(0)$, then the feedback controller \eqref{eq:feedback_1storder} achieves exponential convergence to a bounded ball.
\end{theorem}

\begin{proof}
We have $\dot s = -Ks - \Phi \tilde{\mathbf{a}} + r$ where $\tilde{\mathbf{a}} \coloneqq \hat{\mathbf{a}} - \mathbf{a}$. Considering $V\coloneqq s^\top s + \tilde {\mathbf{a}} ^\top P^{-1} \tilde {\mathbf{a}}$, we again have $\dot V =-2s^\top K s -\tilde {\mathbf{a}}^\top \left( P^{-1}QP^{-1} + \Phi^\top R^{-1}\Phi \right)\tilde{\mathbf{a}} + 2 s^\top r + 2 \tilde {\mathbf{a}}^\top P^{-1}(A\mathbf{a} + b - \dot{\mathbf{a}}) + 2\tilde{\mathbf{a}}^\top\Phi^\top R^{-1} (r+\epsilon)$. The rest of the proof is similar to the proof of Theorem \ref{thm:2ndorder} once one notes that we have $K \succeq \alpha I$ and $Q \succeq 2 \alpha P$ for all small $\alpha>0$, which leads to
\begin{equation}
    \|s(t)\| \le \frac{\overline {d}}{c^* \alpha} + \left(\sqrt{V(0)} - \frac{\overline {d}}{c^*\alpha}\right) e^{-\alpha t},
\end{equation}
where $c^* = \left(1 \wedge \left( \|P_0\|+\| Q\|/(2\lambda)\right)^{-1} \right)^{1/2}$ and $\overline{d}$ is defined as in the proof of Theorem \ref{thm:2ndorder}.
\end{proof}

\section{NEURAL NETWORK ARCHITECTURE AND ITS TRAINING}
\label{sec:nn_arch_training}
In this section, we present a \emph{statistically principled} training procedure for a \emph{uniformly contractive} disturbance predictor whose latent dynamics are governed by feature-dependent matrix coefficients.

\subsection{Contractive Network Parametrization}
All maps are parameterized by GELU MLPs, with spectral normalization applied to the linear layers. To ensure \emph{contraction} of the latent dynamics \eqref{eq:latent_dyn}, we parameterize $A(\phi,u)=S(\phi,u)+K(\phi,u)$ with $S(\phi,u)=-(\beta+\mu(\phi,u))I - L(\phi,u)L(\phi,u)^\top$ and $K(\phi,u) = \frac{1}{2} \left(M(\phi,u)-M(\phi,u)^\top\right)$, with $\beta>0$ fixed. We note that this contractive-by-construction parameterization is inspired by \cite{beikmohammadi2024neuralcontractivedynamicalsystems}. The scalar $\mu(\phi,u)\ge 0$ and the matrices $L(\phi,u)$, $M(\phi,u)$ are produced by MLPs and bounded element-wise via $\sigma(\cdot)$ and $\tanh(\cdot)$, respectively; $L(\phi,u)$ is lower-triangular and $M(\phi,u)$ is strictly lower-triangular. Since $K$ is skew-symmetric, it does not affect the symmetric part of $A$, and we have the following uniform contraction condition:
\begin{equation}
\left(A(\phi,u)+A(\phi,u)^\top\right)/2 =S(\phi,u)\preceq -\beta I,
\end{equation}
therefore the latent dynamics are contracting. $b(\phi,u)$, $\Phi(\phi)$, and $\psi(\phi)$ are parameterized by MLPs; $\Phi$ is bounded via $\tanh(\cdot)$. One may optionally enforce anchoring constraints such as $b(0,0)=0$ or $\psi(0)=0$ by subtracting the value at the origin.
\begin{algorithm}[h]
\caption{Alternating training (hard-EM)}
\label{alg:alt_train}
\begin{algorithmic}[1]
\Require initial $(\theta,R_d,Q_d)$, choice of $\lambda_a$, data windows
\While{stopping criterion is not met}
    \State sample mini-batch $\mathcal B$
    \ForAll{$w\in\mathcal B$}
        \State instantiate the model for window $w$ from $\theta$
        \State apply Kalman smoothing to infer $a_{0:L-1}$
    \EndFor
    \State update $R_d$ (from output residuals; sparse/EMA)
    \State update $Q_d$ (from latent residuals; sparse/EMA)$^*$
    \State update $\theta$ by AdamW with the inferred latents fixed
\EndWhile
\end{algorithmic}
{\footnotesize $^*$optional}
\end{algorithm}

\subsection{Training via Hard Expectation--Maximization}
\label{subsec:hardem}
We train the architecture from data windows $\{(\phi_k,u_k,y_k)\}_{k=0}^{L-1}$ sampled at time step $\Delta t$; windows are extracted from trajectories with a fixed stride and may overlap.

A key challenge in fitting \eqref{eq:latent_model} from data windows is that the latent initial condition $\mathbf a_0=\mathbf a(0)$ is \emph{unobserved}, \emph{window-specific}, and has no direct physical meaning. As a result, $\mathbf a_0$ is a missing value that must be inferred jointly with the shared parameters, which can create an identifiability ambiguity between $\mathbf a_0$ and $\theta$. To regularize this inference problem in a principled way, we place a zero-mean Gaussian prior on the initial latent state, $\mathbf a_0 \sim \mathcal N(0,\lambda_a^{-1}I)$, which induces the quadratic penalty $\lambda_a\|\mathbf a_0\|^2$ in the training objective. Given parameters $\theta$, we roll out \eqref{eq:latent_dyn} by RK2 to obtain $\mathbf a_k$, form predictions $\hat d_k=\Phi_\theta(\phi_k)\mathbf a_k+\psi_\theta(\phi_k)$, and assume a Gaussian observation model $y_k \mid \mathbf a_0,\theta \sim \mathcal N(\hat d_k, R_d)$ with $R_d \succ 0$. This yields a Gaussian negative log-likelihood (NLL) loss
\begin{equation}
\label{eq:train_obj}
\begin{split}
\min_{\mathbf a_0,\theta,R_d}\;& \sum_{k=0}^{L-1}\Big((y_k-\hat d_k(\mathbf a_0,\theta))^\top R_d^{-1}(y_k-\hat d_k(\mathbf a_0,\theta)) \\
&\qquad\qquad\qquad +\log|R_d| \Big)+\lambda_a\|\mathbf a_0\|^2 ,
\end{split}
\end{equation}
implicitly averaged over a mini-batch of windows. We optimize \eqref{eq:train_obj} by an alternating scheme, which can be interpreted as a (generalized) hard-EM procedure \cite{CeGo1992}:
\emph{(hard E-step)} fix $\theta$ and $R_d$ and perform exact minimization with respect to $\mathbf a_0$ to obtain a MAP estimate for each window;
\emph{(M-step)} Fix the inferred $\mathbf a_0$ and update $R_d$ from mini-batch statistics; then use AdamW to update $\theta$. The quadratic prior term $\lambda_a\|\mathbf a_0\|^2$ regularizes the inferred initial condition and could be interpreted as mitigating the identifiability ambiguity.

For fixed $\theta$, the latent rollout $\{\mathbf{a}_k\}$ depends affinely on $\mathbf{a}_0$; consequently, $\hat d_k$ is affine in $\mathbf{a}_0$ and the Gaussian NLL induces a positive definite quadratic objective in $\mathbf{a}_0$. Thus, the $\mathbf{a}_0$ block can be optimized exactly, and this optimization can be viewed as a special case of Kalman smoothing; see Section~\ref{subsec:process_noise_kalman}. In view of \eqref{eq:train_obj}, it is natural to use $P_0 = \lambda_a^{-1} I$, $\hat {\mathbf{a}}(0)=0$, $R =\left(\Delta t \right) R_d $ and a small $Q = qI\succ0$ in~\eqref{eq:KalmanComposite}.

\subsection{Process Noise Modeling and the Kalman Smoother}
\label{subsec:process_noise_kalman}
The training scheme in the previous subsection corresponds to a \emph{deterministic} latent rollout within each window: once $(\theta,\mathbf a_0)$ are fixed, the entire latent trajectory $\{\mathbf a_k\}_{k=0}^{L-1}$ is fixed. We optionally add discrete-time process noise $Q_d\succ0$ by treating the \emph{entire} latent trajectory as missing data; Section~\ref{subsec:hardem} corresponds to the degenerate special case of $Q_d = 0$. Concretely, let $f_{\theta,k-1}(\mathbf a_{k-1})$ denote the one-step RK2 update of \eqref{eq:latent_dyn} using $(\phi_{k-1},u_{k-1})$ (from time $k-1$ to time $k$). Since \eqref{eq:latent_dyn} is affine in $\mathbf a$, $f_{\theta,k-1}(\mathbf a_{k-1})$ is affine in $\mathbf a_{k-1}$. We replace the deterministic update with
\begin{equation}
\mathbf a_k \mid \mathbf a_{k-1},\theta \sim \mathcal N\!\big(f_{\theta,k-1}(\mathbf a_{k-1}),\,Q_d\big),\qquad k=1,\dots,L-1,
\end{equation}
while keeping the same observation model for $y_k$ and the same prior $\mathbf a_0\sim\mathcal N(0,\lambda_a^{-1}I)$. Compared to \eqref{eq:train_obj}, we replace $\hat d_k(\mathbf a_0,\theta)$ by $\hat d_k(\mathbf a_k,\theta) \coloneqq \Phi_\theta(\phi_k)\mathbf a_k+\psi_\theta(\phi_k)$, and the window NLL loss acquires the additional process-noise term $(L-1)\log|Q_d| +\sum_{k=1}^{L-1}\big(\mathbf a_k-f_{\theta,k-1}(\mathbf a_{k-1})\big)^\top Q_d^{-1}\big(\mathbf a_k-f_{\theta,k-1}(\mathbf a_{k-1})\big)$.

For fixed $(\theta,R_d,Q_d)$, minimization over the latent trajectory $\mathbf a_{0:L-1}$ (the hard E-step) is exactly solvable via a Kalman smoother \cite{SaSt2023}, yielding an inferred latent trajectory for each window. The Kalman smoother is implemented by the RTS formulas. In the M-step, we fix this inferred trajectory and update $R_d$ and optionally $Q_d$ from mini-batch statistics, and then $\theta$ by AdamW (see Algorithm~\ref{alg:alt_train}).

The identified discrete-time covariances lead to the continuous-time gains in~\eqref{eq:KalmanComposite}: $Q \approx Q_d/\Delta t$ and $R \approx (\Delta t) R_d$, as noted in Remark~\ref{remark:data_driven_covariances}. In our implementations, we used diagonal $R_d$ and isotropic $Q_d=q_d I$ for simplicity.

\section{EXPERIMENTAL RESULTS}
\label{sec:exp_results}
Implementing the controllers~\eqref{eq:feedback_2ndorder} and \eqref{eq:feedback_1storder} requires discretizing \eqref{eq:KalmanComposite}.
We use a Joseph-form Kalman filter, with forward Euler discretization for the $P \Phi^\top s$ term. For both examples, we used $d_a = 3$. We compare our controller with \emph{(i)} a model-based PD control loop, \emph{(ii)} the fixed-decay ablations inspired by \cite{OCSh2022, LuXi2025}, and \emph{(iii)} a DAC controller in which an LTI system is identified by \cite[Sections~4.3.1, 4.4.2]{OvMo1996}, a variant of the \textsc{N4SID} algorithm~\cite{VANOVERSCHEE199475}. We denote the last controller as \textsc{N4SID}v-DAC. The fixed-decay ablations, denoted FixedDecay, were trained by the Kalman-smoother-based hard-EM procedure restricted to the function class $A = -\lambda I$, $b\equiv 0$, and $\psi \equiv 0$. For the adaptive controllers, datasets were split into training, validation, and test sets in a 70:15:15 ratio. For the \textsc{N4SID} variant, the number of block rows $i$ in the past and future Hankel matrices, following the notation of \cite{OvMo1996}, was selected to minimize the validation-set MSE. The LTI model was then re-estimated using the combined training and validation sets.

The LTI disturbance augmentations are of the form
\begin{subequations}
\begin{align}
    \dot{\mathbf a} &= A \mathbf a + B_1 \phi + B_2 u,\\
    \hat d &= C \mathbf a + D_1 \phi + D_2 u, \label{eq:lti_sysid_output}
\end{align}
\end{subequations}
with $d \approx \hat d$. For example, substituting \eqref{eq:lti_sysid_output} into \eqref{eq:GVR_Bot} yields $\dot v \approx A_v v + (B + D_2)u + D_1 \phi + C\mathbf a$. Therefore, considering \eqref{eq:feedback_1storder}, the disturbance-compensating command is obtained by solving the equation $(B+D_2)u+D_1 \phi + C\hat{\mathbf a}=\dot v_d-A_v v -Ks$ in $u$. Although the LTI model is presented in continuous-time notation, system identification was performed in discrete time.

\subsection{A Slippery Tracked Vehicle with a Partially Filled Liquid Tank and a Swinging Pendulum}

The presented architecture was evaluated on a tracked ground vehicle \censor{(``GVR-Bot'')} carrying a pendulum and a liquid tank filled with water to approximately 30\% capacity. The experiment also provides a test of generalization to an unseen intermediate fill level because the 30\% level was absent from the training data. The choice of a liquid-carrying platform was partly motivated by the growing interest in on-orbit refueling~\cite{RaAm2026}. To emulate low-traction, slippery terrain, the tracks were wrapped with low-friction tape. Water-tank level sensor measurements $w\in\mathbb{R}$ were sampled at 18 Hz. The platform had an NVIDIA Jetson AGX Orin as a companion computer, and used visual--inertial odometry (VIO) for localization. The control loop ran at 50 Hz.

The vehicle's velocity is modeled by the following first-order lag \cite{LuXi2025}:
\begin{equation}
\label{eq:GVR_Bot}
    \dot v = A_v v + B u + d, 
\end{equation}
where $v = [v_x, \omega_z]$ denotes the vehicle's forward velocity and yaw rate, $u = [v_{x, \text{cmd}}, \omega_{z,\text{cmd}}]$ denotes the commanded values, and $B = -A_v=\operatorname{diag}(1/0.2, 1/0.15)$, determined by system identification on hardware. The vehicle's internal PID controller is inaccessible, and only the velocity commands are available, which motivates the model \eqref{eq:GVR_Bot}. The liquid sloshing and pendulum motion are expected to contribute to time-varying changes in the vehicle's traction/slip behavior and, consequently, in the observed command-to-velocity response. Such unmodeled effects are lumped into $d$ in \eqref{eq:GVR_Bot}. 

\paragraph{Data Collection and Training}
Data were collected for about 11 minutes for each water-tank level ($\approx$ 60\%, 0\%) in an open testing area, with the vehicle driven at various speeds and in various directions without following a prescribed trajectory. The vehicle was driven through the full operational range of speeds and angular rates, with maneuvers such as abrupt starts and stops or smooth accelerations. We trained the neural network ($\approx 4.9\mathrm{k}$ trainable parameters) using $(v, w, u, y)$ data with $\Delta t = 0.01$. Here, $y$ is the estimated $d$ under the zero-order hold assumption, and we used $\phi = (v,w)$. For both the fixed-decay ablation and the proposed method, the neural networks were trained for 100 epochs using a learning rate of 0.001 and a weight decay of 0.01. The training took less than 3 minutes on an Apple M4 Pro laptop.

\paragraph{Predictions for Disturbance $d$} 
The predictive capability of the architecture was evaluated by rolling out the dynamics in~\eqref{eq:latent_dyn} and \eqref{eq:affine_readout} using the input, state, and water-level measurements from the test set. Figure~\ref{fig:prediction_gvr} illustrates the predicted and measured disturbances on two test-set trajectories.

\begin{figure}[!ht]
\vspace{-9pt}
\centering
\setlength{\tabcolsep}{1pt}
\renewcommand{\arraystretch}{0}
\begin{tabular}{@{}cc@{}}
\includegraphics[width=0.495\linewidth]{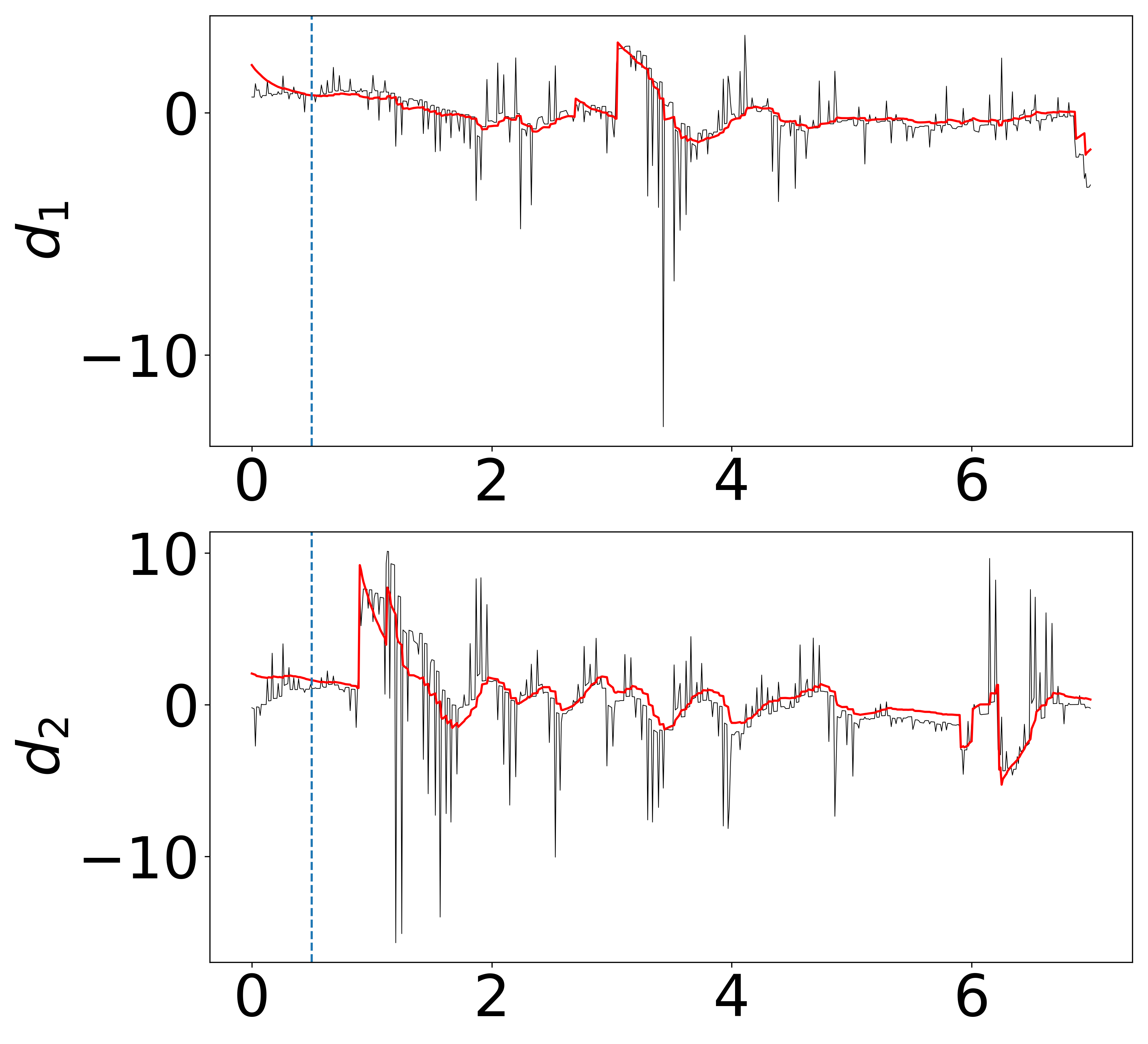} &
\includegraphics[width=0.495\linewidth]{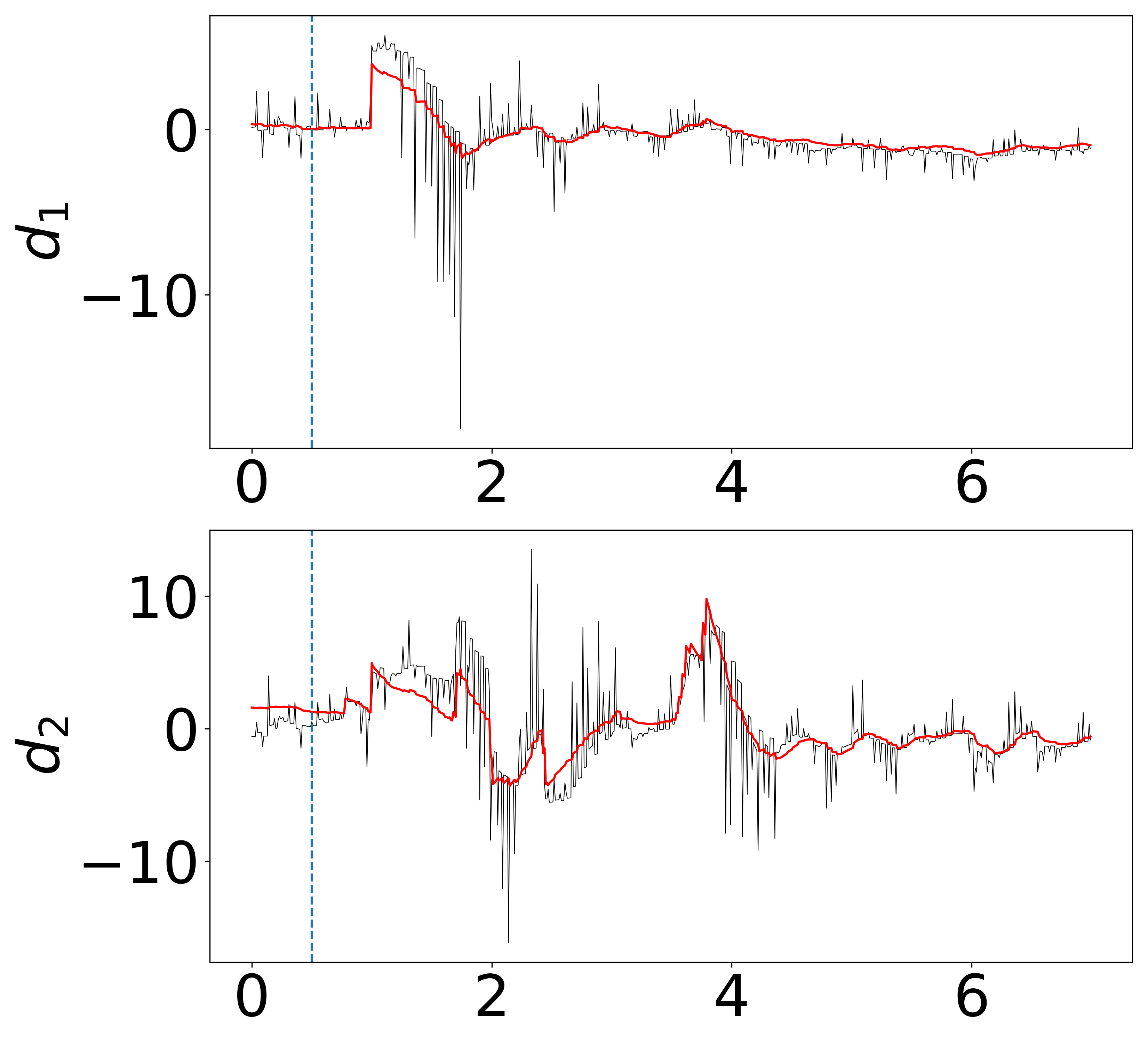}
\end{tabular}
\caption{(Tracked vehicle) Test-set rollouts for the predicted disturbance $\hat d \in \mathbb{R}^2$ as functions of time, on two different test-set trajectories. Black: disturbance proxy $y(t)$; red: rollout prediction $\hat d(t)$. The proxy $y$ is visibly noisier. The average prediction MSE was $3.53$; it was $3.49$ after discarding the initial half second for each trajectory as a burn-in period (vertical dashed line).}
\label{fig:prediction_gvr}
\vspace{-10pt}
\end{figure}
\paragraph{Tracking Controller Test}
We compare the proposed controller with three other controllers: our FixedDecay ablation, \textsc{N4SID}v-DAC, and the model-based PD loop used in~\cite[Eqs.~14--18]{LuXi2025} which contains nonlinearities. The vehicle's commanded trajectory consisted of two consecutive double lane changes, which excited the water and pendulum dynamics while inducing track slip.

The model-based PD baseline already provides strong overall tracking, particularly in the forward-velocity channel. The benefits of incorporating disturbance predictors are most apparent in yaw-rate regulation. As shown in Table~\ref{tab:gvr_tracking_results}, the proposed controller achieves the lowest angular-rate and velocity RMS errors and ties \textsc{N4SID}v-DAC for the lowest position RMS error. The competitive angular-rate and position tracking performance of \textsc{N4SID}v-DAC suggests that an LTI augmentation captures a significant part of the model mismatch over the tested regime. In contrast, the fixed-decay ablation performed less favorably in this experiment.

\begin{figure}[!h]
\centering
\setlength{\tabcolsep}{1pt}
\renewcommand{\arraystretch}{0}
\includegraphics[width=1\columnwidth]{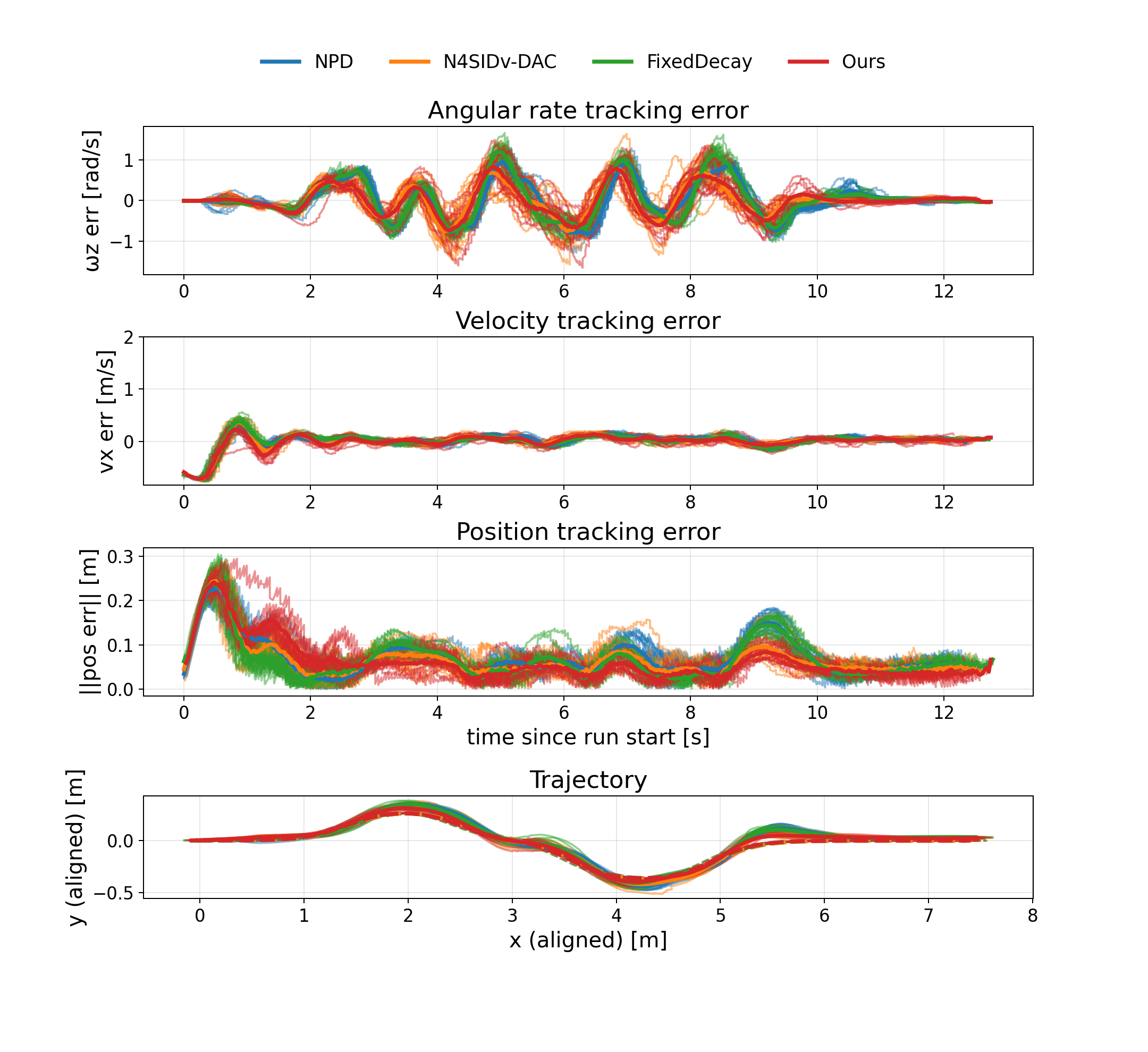}
\vspace{-40pt}    
\caption{(Tracked vehicle) Tracking runs under four controllers.}
\label{fig:tracking_gvr}
\vspace{-15pt}
\end{figure}

\begin{table}[!hb]
\vspace{-5pt}
\centering
\small
\setlength{\tabcolsep}{4pt}
\resizebox{\columnwidth}{!}{%
    \begin{tabular}{lcccc}
    \toprule
    Metric & NPD & \textsc{N4SID}v-DAC & FixedDecay & Ours \\
    \midrule
    Ang. rate error RMS [rad/s] & 0.449 & 0.414 & 0.468 & \textbf{0.407} \\
    Vel. error RMS [m/s] & 0.151 & 0.160 & 0.156 & \textbf{0.149} \\
    Pos. error RMS [m] & 0.085 & \textbf{0.078} & 0.081 & \textbf{0.078} \\
    \bottomrule
    \end{tabular}%
}
\caption{Tracked-vehicle RMS tracking errors.\\
{\normalfont NPD: nonlinear PD loop; \textsc{N4SID}v-DAC: LTI-based disturbance-accommodating control; FixedDecay: fixed-decay ablation.}}
\vspace{-20pt}
\label{tab:gvr_tracking_results}
\end{table}

\subsection{Coupled Duffing Oscillators: a Robustness Sweep}

We consider a pair of Duffing oscillators coupled by a nonlinear spring:
\begin{subequations}
\begin{align}
  m_1\ddot x_1 &= - k_1 x_1 - \alpha_1 x_1^3 +u +(d-c_1 \dot x_1),\\[-3pt]
  m_2\ddot x_2 &= - k_2 x_2 - \alpha_2 x_2^3-c_2 \dot x_2  - d,\\[-3pt]
  d &= k_c (x_2 - x_1) + \alpha_c (x_2 -x_1)^3,
\end{align}
\end{subequations}
where we only observe $x_1,\dot x_1,u$ and are unaware of the existence of the coupled $x_2$-system. We train the neural networks with a synthetic dataset of $(x_1, \dot x_1, u, y)$ ($\Delta t =0.01$). Here, the samples of $x_1$ and $\dot x_1$ are corrupted by zero-mean Gaussian noise with standard deviation $0.005$; $y$ is computed using forward differences of the sampled $\dot x_1$. The parameters used for the nonlinear oscillators were:
\begin{subequations}
\begin{align}
m_1 &= 1.0, & m_2 &= 0.7, \\[-3pt]
c_1 &= 0.30, & c_2 &= 0.25, \\[-3pt]
k_1 &= 1.0, & k_2 &= 0.8, & k_c &= 2.0, \\[-3pt]
\alpha_1 &= 1.7, & \alpha_2 &= 0.2, & \alpha_c &= 5.0.
\end{align}
\end{subequations}
Our empirical evaluation considers two complementary aspects: \emph{(a)} we show that \eqref{eq:latent_model} with $\phi=(x_1, \dot x_1)$ predicts $d$ (or more precisely, its proxy $y$) accurately on the test set, and \emph{(b)} we compare the controller \eqref{eq:feedback_2ndorder} with the baseline model-based PD ($u = m_1 \dot v_r + k_1 x_1 + \alpha_1 x_1^3 - Ks + c_1 \dot x_1$), \textsc{N4SID}v-DAC, and the fixed-decay representation ablation. 

\paragraph{Predictions for Disturbance $d$}
We rolled out~\eqref{eq:latent_model} starting from $\mathbf{a}_0 = 0$ to get the estimates $\hat d$, and compared them to the test-set labels $y$: see Fig.~\ref{fig:prediction_duffing}. Note that the initial transients from the fixed initial condition $\mathbf{a}_0 = 0$ decay as time passes, as expected from the contraction property (Proposition~\ref{prop:insensitivity_IC}). This predictive capability arises from generalizing prior ``fixed-decay'' representation-learning approaches to a more flexible learned dynamical model.

\begin{figure}[t!]
\vspace{10pt}
\centering
\setlength{\tabcolsep}{1pt}
\renewcommand{\arraystretch}{0}
\begin{tabular}{@{}cc@{}}
\includegraphics[width=0.495\linewidth]{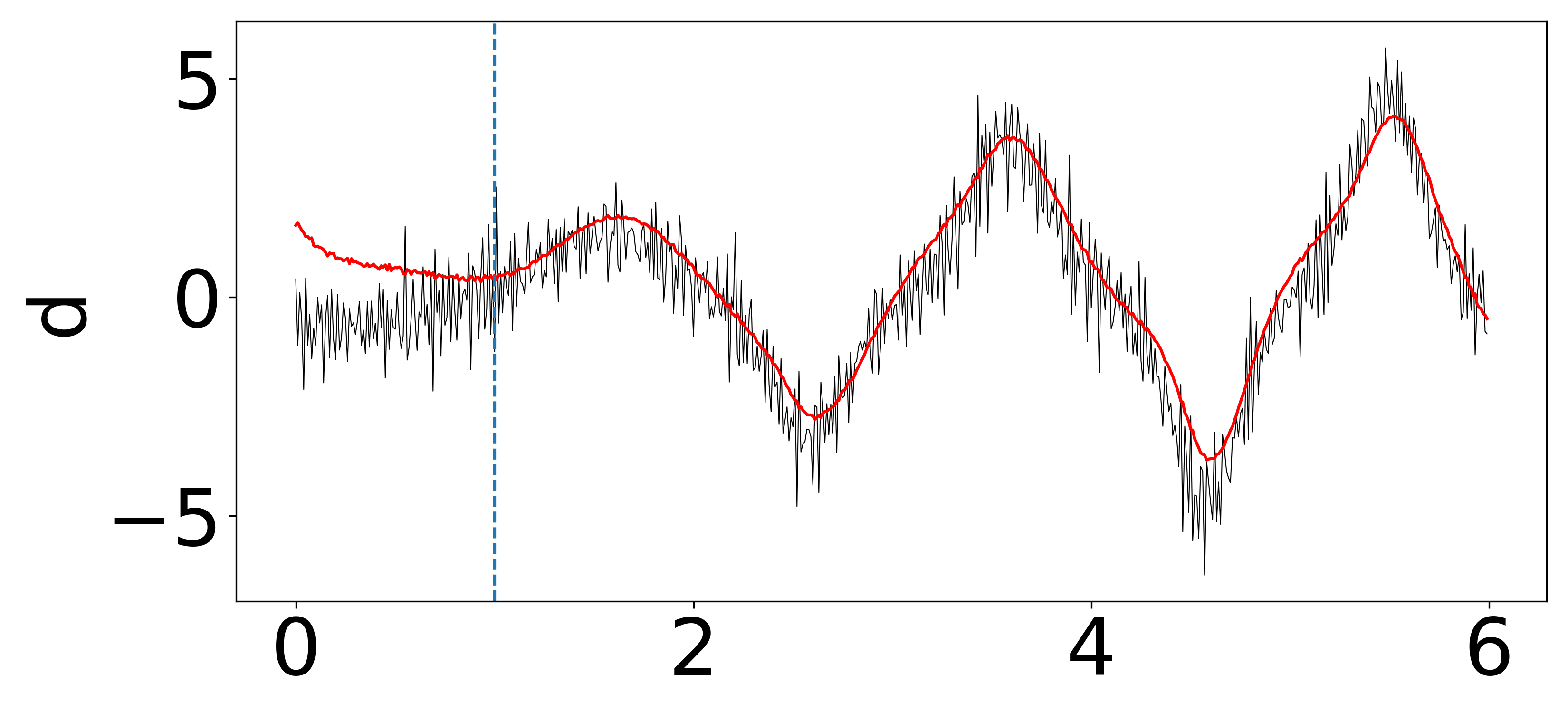} &
\includegraphics[width=0.495\linewidth]{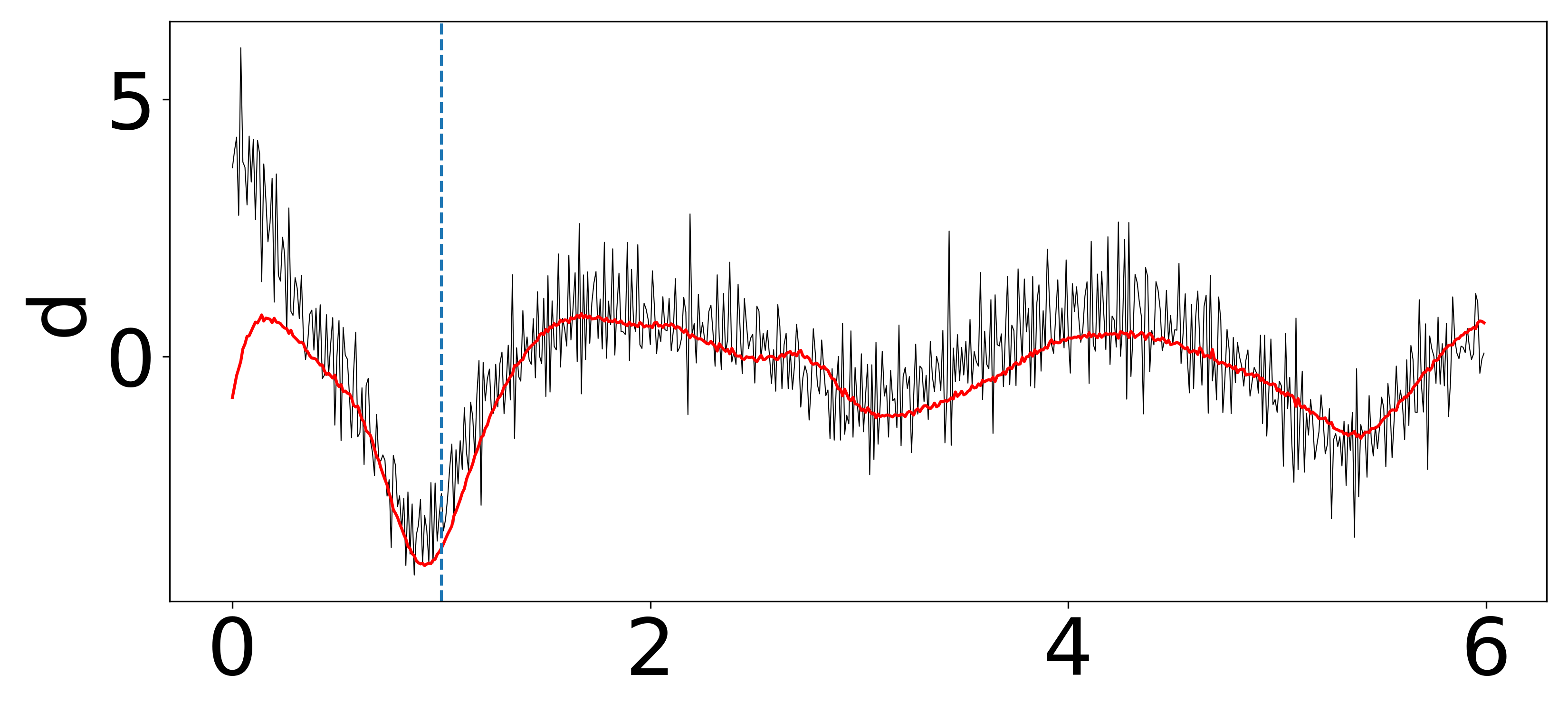} \\
\includegraphics[width=0.495\linewidth]{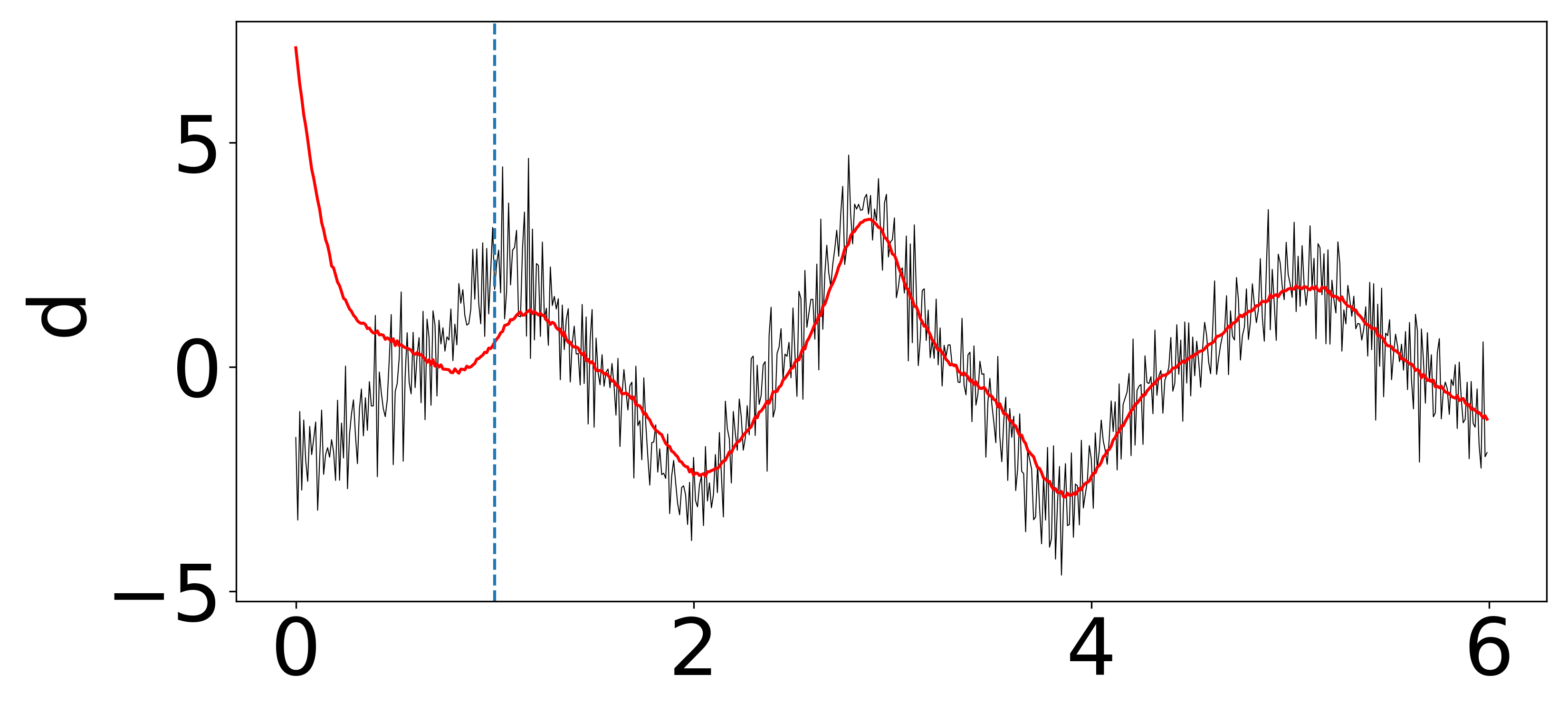} &
\includegraphics[width=0.495\linewidth]{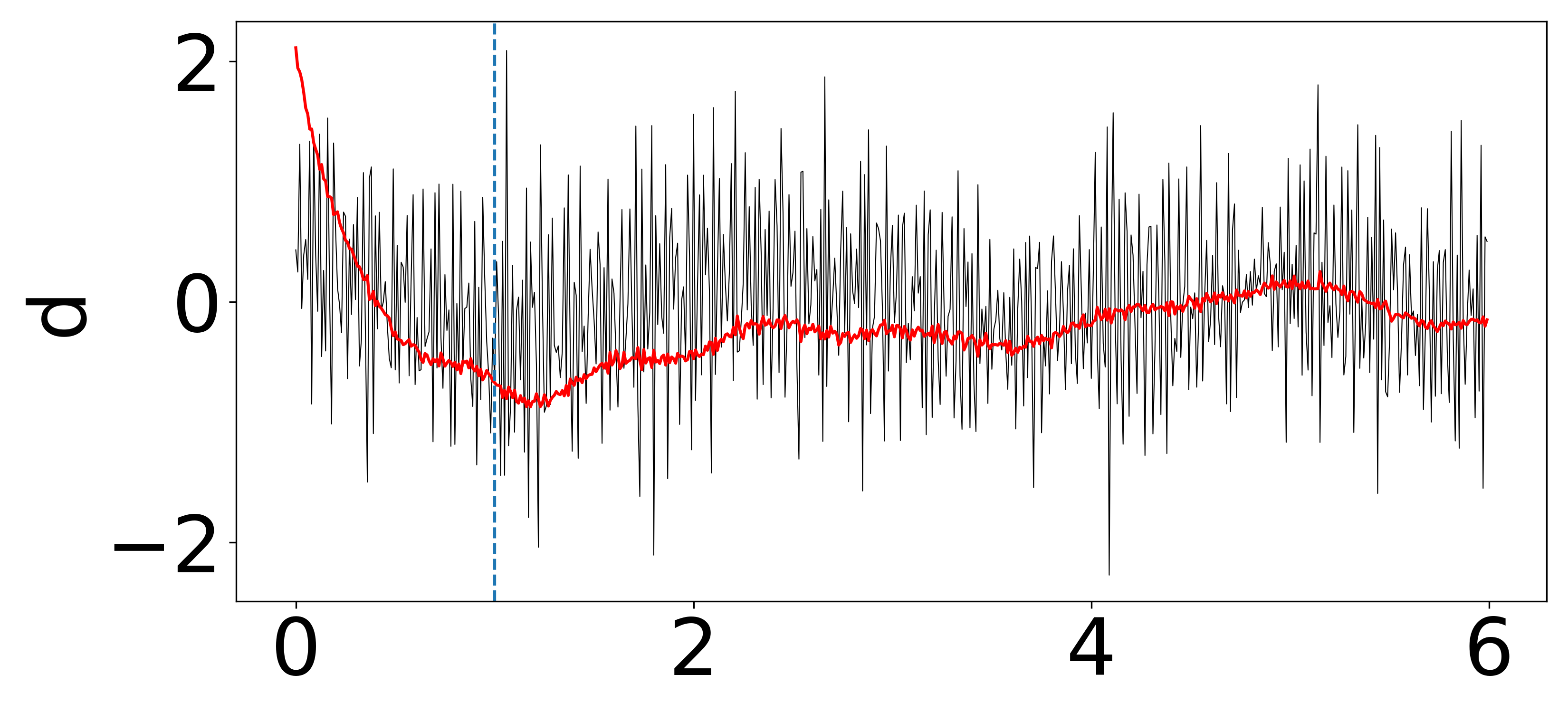}
\end{tabular}
\vspace{-10pt}
\caption{(Coupled nonlinear oscillator) Test-set rollouts from Eq.~\eqref{eq:latent_model} for the predicted disturbance as functions of time. Black: disturbance proxy $y(t)$; red: rollout prediction $\hat d(t)$. The proxy $y$ is visibly noisier, while $\hat d$ is smooth. Mean squared error is $2.54$ overall and $0.88$ after discarding the first second of each trajectory as a burn-in period (vertical dashed line).}
\label{fig:prediction_duffing}
\end{figure}

\paragraph{Tracking Controller Test}
To illustrate the benefit of predictive capability, we made $y$ available at only 1 Hz. We also injected different levels of Gaussian noise (its standard deviation denoted by $\sigma$) into $y$ to test for robustness (Fig.~\ref{fig:robustness_sweep}). We used $\Lambda = 2.0$ and $K = 12.0$ for all four controllers.
\begin{figure} [h!]
\centering
\includegraphics[width=0.95\linewidth]{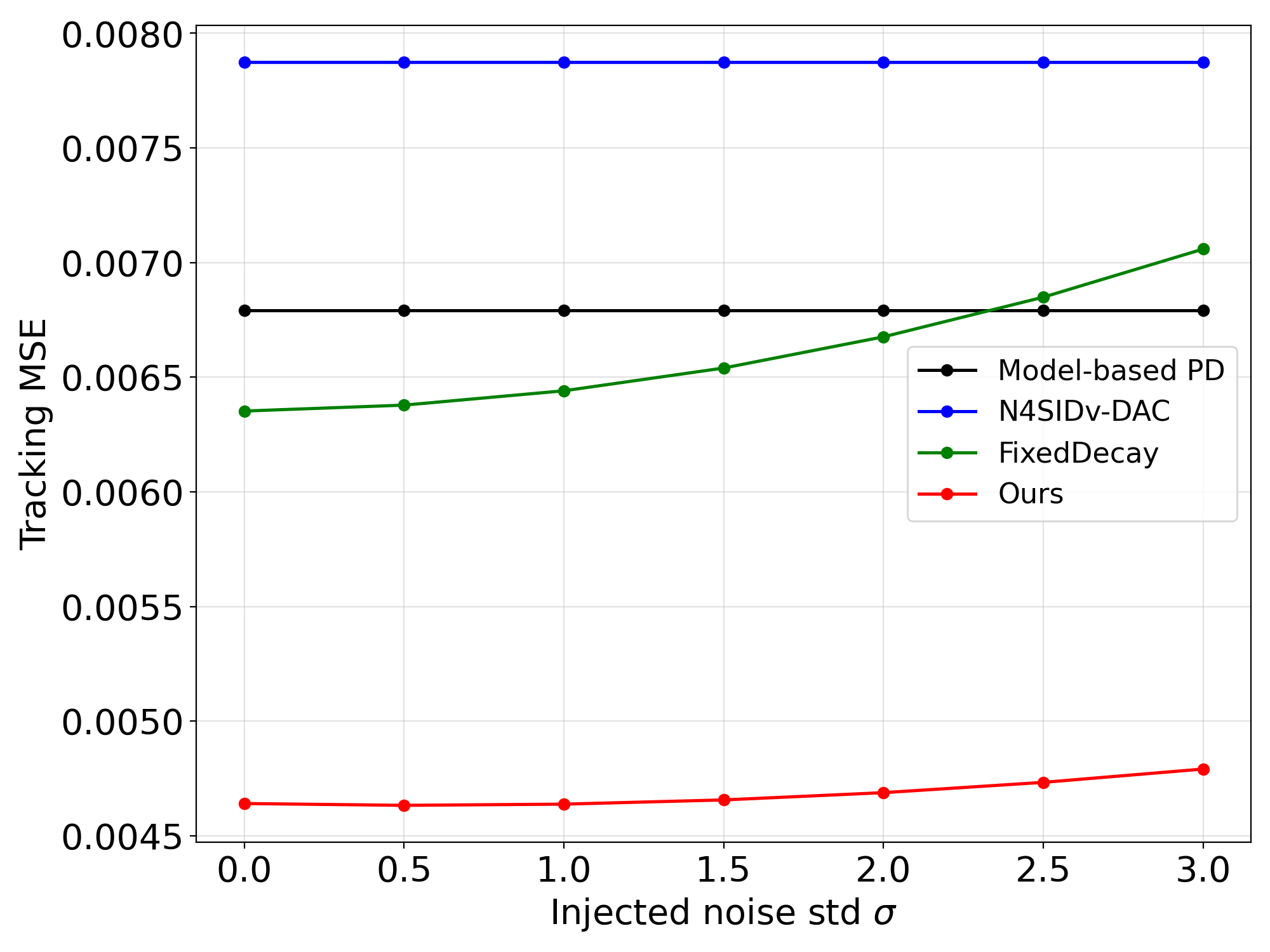}
\vspace{-10pt}
\caption{(Coupled nonlinear oscillator) Robustness sweep for tracking control under additional Gaussian noise injected into the sparsely provided disturbance proxy $y$ (averaged over 30 reference trajectories, with each trajectory being 20 seconds long). Mean tracking MSE is plotted versus the added noise level $\sigma$, where $\sigma=0$ denotes the baseline noise setting. The proposed dynamical-representation controller performs the best. As $\sigma$ increases, the neural adaptive controllers degrade because the disturbance proxy becomes increasingly corrupted. The simpler \textsc{N4SID}v-DAC baseline resulted in a flat graph because the fitted discrete-time $A$, $B_1$, and $B_2$ satisfied $\left\|A\right\|\approx 0.09 \ll 1$, $\left\|B_1\right\|\approx 0.08 \ll 1$, $\left\|B_2\right\|\approx 0.005 \ll 1$ and hence the disturbance compensation mostly came from $\hat d \approx D_1 \phi + D_2 u$.}
\label{fig:robustness_sweep}
\vspace{-20pt}
\end{figure}

In this experiment, the acceleration $\ddot{x}_1$ was estimated by taking backward finite differences of the observed $\dot x_1$, and then filtered by an exponential moving average. The disturbance proxy $y$ was then calculated as $m_1 \ddot x_1 +c_1 \dot x_1 +k_1 x_1 + \alpha_1 x_1^3-u$; after this, it was corrupted by noise level $\sigma$. Note that $\sigma=0$ already results in a noisy proxy $y$ because the measurements for $x_1$ and $\dot x_1$ are noisy. Each reference trajectory was generated as a sum of randomized sinusoidal signals. 
 
Figure~\ref{fig:robustness_sweep} summarizes the robustness sweep for the proposed dynamical-representation controller, the fixed-decay representation ablation, the \textsc{N4SID}v-DAC controller, and the model-based PD baseline. In the nominal-noise case ($\sigma = 0$), the proposed controller achieves the lowest average tracking MSE; in particular, our controller reduced tracking MSE by 27\% compared to the closest competitor. In this highly nonlinear setting, \textsc{N4SID}v-DAC performed worse than model-based PD, suggesting that the identified LTI disturbance model was inadequate. Overall, the proposed method provides the strongest performance and remains strong throughout the sweep.

\section{CONCLUSIONS}
This paper presented a composite adaptive control framework based on a representation-learning approach for systems with unobserved (or partially observed), dynamically coupled disturbance sources. The presented controller complements prior ``fixed-decay'' last-layer representation-learning approaches by incorporating the classical disturbance-accommodating control (DAC) viewpoint. By modeling the disturbance as the output of a uniformly contractive latent dynamical system, the proposed method introduces a dynamics prior that supports disturbance prediction while mitigating sensitivity to latent-state initialization errors. The latent dynamics are parameterized by neural networks and trained using a statistically principled hard expectation--maximization procedure.

Hardware experiments on a tracked vehicle carrying a partially filled liquid tank and a pendulum illustrate the practical feasibility of the proposed disturbance estimator and the composite adaptive controller. The learned dynamical representation improves overall tracking performance during aggressive maneuvering, with the clearest improvement observed in yaw-rate regulation. In simulation, experiments on a system of coupled Duffing oscillators demonstrate accurate disturbance prediction and improved tracking relative to the baselines, while remaining strong under progressively noisier disturbance proxies $y$. Here, the sparse availability of $y$ allowed us to assess the role of predictive capability for maintaining tracking performance.

Taken together, the results indicate that the proposed method holds substantial promise for real-world deployment, particularly when disturbance proxies are sampled at intervals which are not negligible relative to the characteristic disturbance timescale, making predictive compensation important between observations. Future work may include evaluating the method across a wider range of hardware platforms and experimental conditions, as well as systematically studying the training dynamics of the hard-EM procedure, particularly the evolution of covariance estimates. Future work may also include replacing the present hard E-step with the standard EM E-step, which takes expectations with respect to the posterior distribution over latent trajectories. The proposed dynamical representation could also serve as a disturbance model in stochastic MPC with chance constraints \cite{Me2016}, potentially useful for real-time trajectory \emph{planning}.

{\footnotesize
\noindent\emph{Acknowledgment:} The authors thank G. X. Johnson, J. Alindogan, and S. Sukhatme for their help with hardware troubleshooting. The authors also thank the anonymous reviewers.
\par}
\vspace{-8pt}

\bibliographystyle{IEEEtran}
\bibliography{IEEEabrv,references}

\section*{APPENDIX}
\subsection{Hard and Soft Expectation--Maximization}
We briefly outline the changes required to replace hard EM with soft EM.
\paragraph{Hard vs. Soft EM}
Suppose a window of measured features, control inputs, and disturbance observations $\{(\phi_k,u_k,y_k)\}_{k=0}^{L-1}$ is given. Write $\mathbf{y} \coloneqq y_{0:L-1}$ and $\eta \coloneqq (\theta,R_d,Q_d)$. In the general case, the missing value is the entire trajectory $Z=\mathbf a_{0:L-1}$. For a fixed $\eta$, the statistical model in Section~\ref{subsec:process_noise_kalman} can be written as
\begin{align*}
\mathbf a_{k+1}&=F_k\mathbf a_k+g_k+\xi_k, & \xi_k&\sim\mathcal{N} (0,Q_d),\\
y_k&=H_k\mathbf a_k+h_k+\epsilon_k,&\epsilon_k&\sim\mathcal{N} (0,R_d),
\end{align*}
with $\mathbf a_0 \sim \mathcal{N} (0,\lambda_a^{-1}I)$, where $\lambda_a>0$ is a fixed hyperparameter. Here, $\mathbf a_0$, $\xi_k$, $\epsilon_k$ are independent. Note that $(F_k,g_k)$ comes from discretizing \eqref{eq:latent_dyn} in time by, say, RK2 or forward Euler; the coefficients $(F_k, g_k, H_k, h_k)$ depend on $\theta$.

At iteration $r$, let $\pi^r(z) \coloneqq p(z\mid\mathbf{y},\eta^r)$ and $Q^r(\eta) \coloneqq -\mathbb E_{\pi^r} \left[\log p(\mathbf{y},Z \mid \eta)\right]$. Standard (soft) EM minimizes $Q^r(\eta)$ with respect to $\eta$, holding $\pi^r$ fixed, to obtain $\eta^{r+1}$. Hard EM instead selects $\hat z^r \coloneqq \arg\max_z \pi^r(z)$ and minimizes $-\log p(\mathbf y,\hat z^r\mid\eta)$ with respect to $\eta$, holding $\hat z^r$ fixed.
In other words, soft EM averages over the latent uncertainty $Z\sim\pi^r$, whereas hard EM uses a single point estimate of $Z$, replacing the soft E-step (expectation) with a hard E-step (posterior maximization).

\paragraph{The Soft EM Loss}
Denote the current parameters as $\eta^r=(\theta^r,R_d^r,Q_d^r)$, where we assume $R_d^r \succ 0$ and $Q_d^r \succ 0$. Compute the \emph{smoothed} means and covariances using a Kalman smoother under parameters $\eta^r$, and denote the smoothed moments by $m_k=\mathbb E_{\pi^r}[\mathbf a_k]$, $P_k=\operatorname{Cov}_{\pi^r}(\mathbf a_k)$, and $C_k=\operatorname{Cov}_{\pi^r}(\mathbf a_{k+1},\mathbf a_k)$; here, the cross-covariances are obtained from the smoothing gains and the smoothed covariances as follows~\cite[Eq.~(12.13)]{SaSv2023}:
\begin{equation}
G_k \coloneqq P_{k|k}F_k^\top(P_{k+1|k})^{-1}, \qquad C_k=P_{k+1}G_k^\top.
\end{equation}
The posterior moments $(m_k,P_k,C_k)$ remain fixed throughout the M-step.

For candidate parameters $\eta$, let $\mathbf a_{0:L-1}\sim\pi^r$ and define the random variables $r_k^y \coloneqq y_k - H_k \mathbf{a}_k - h_k$ and $r_k^a \coloneqq \mathbf a_{k+1} - F_k \mathbf{a}_k - g_k$. We have
\begin{align*}
    r^y_k &= e_k - H_k (\mathbf a _k -m_k),\\
    r^a_k &= v_k + (\mathbf a_{k+1}-m_{k+1}) - F_k (\mathbf a _k -m_k),
\end{align*}
where $e_k\coloneqq y_k-H_k m_k-h_k$ and $v_k\coloneqq m_{k+1}-F_k m_k-g_k$. Therefore, the residual second moments are
\begin{subequations}
\begin{align}
S^y_k&\coloneqq \mathbb{E}_{\pi^r}[r_k^y (r_k^y)^\top]=e_k e_k^\top+H_k P_k H_k^\top,\\
\begin{split}
S^a_k&\coloneqq \mathbb{E}_{\pi^r}[r_k^a (r_k^a)^\top]=v_k v_k^\top+P_{k+1}+F_k P_k F_k^\top\\&\qquad\qquad\qquad\qquad\qquad-C_k F_k^\top-F_k C_k^\top.
\end{split}
\end{align}
\end{subequations}
From this, we can calculate the soft-EM loss averaged over a mini-batch of $B$ windows (indexed by $w=1,\ldots,B$) as
\begin{equation}
\label{eq:appen_em_loss}
\begin{split}
Q&^r(\eta)=\frac 1 B \sum_{w=1}^B\frac12\Biggl[\sum_{k=0}^{L-1} \left[\operatorname{tr}(R_d^{-1}S^{y,(w)}_k)+\log|R_d|\right]\\
&+\sum_{k=0}^{L-2} \left[\operatorname{tr}(Q_d^{-1}S^{a,(w)}_k)+\log|Q_d|\right]\\
&\!\!\!\!\!\!+\lambda_a \! \left(\!\|m_0^{(w)}\|^2\!+\!\operatorname{tr}P_0^{(w)}\!\right)\!\!+\!L(d_a+d_u)\log(2\pi)\!-\!d_a \log\lambda_a \!\Biggr]\!.
\end{split}
\end{equation}
Note that the last line in \eqref{eq:appen_em_loss} can be omitted in the loss function for the M-step.

Now, our M-step is as follows: compute the new $(R_d, Q_d)$ by
\begin{subequations}
\begin{align}
R_d^+&=\frac{1}{BL} \sum_{w=1}^{B}\sum_{k=0}^{L-1}S^{y,(w)}_k,\\
Q_d^+&=\frac{1}{B(L-1)} \sum_{w=1}^{B}\sum_{k=0}^{L-2}S^{a,(w)}_k,
\end{align}
\end{subequations}
and then update $\theta$ using an optimizer such as AdamW to minimize the loss in \eqref{eq:appen_em_loss}, omitting its last line. In practice, eigenvalue clipping might be needed to ensure positive definiteness of $(R_d^+, Q_d^+)$. 

\paragraph{The $Q_d=0$ Variant}
In Section~\ref{subsec:hardem}, the latent trajectory is determined by $\theta$ and the initial state $\mathbf a_0$, so we take the missing value as $Z=\mathbf a_0$. Hard EM uses the MAP estimate for $\mathbf a_0$; soft EM instead retains the posterior distribution $\mathbf a_0\mid\mathbf y,\eta^r\sim\mathcal N(m_0,P_0)$. The posterior moments $(m_0,P_0)$ remain fixed throughout the soft-EM M-step.

Write the deterministic rollout-plus-readout map as $\hat d_k=M_k\mathbf a_0+c_k$ where $(M_k,c_k)$ depend on the candidate $\theta$, and define
\begin{align}
e_k&\coloneqq y_k-M_km_0-c_k,\\
S^y_k&\coloneqq e_ke_k^\top+M_kP_0M_k^\top.
\end{align}
Assuming $R_d \succ 0$, the soft-EM objective reduces to
\begin{equation}
Q^r(\eta)=\frac 1 B \sum_{w=1}^{B}\frac 1 2 \sum_{k=0}^{L-1}\left[\operatorname{tr}(R_d^{-1}S_k^{y,(w)})+\log|R_d|\right]
\end{equation}
up to additive constants.

\subsection{Regularization in the \textsc{N4SID}-Variant Implementation}
Relative to the algorithm in~\cite[Fig.~4.7]{OvMo1996}, our implementation applies ridge-regularized pseudoinverse approximations at three locations: \emph{(i)} in the oblique projections~\cite[Eq.~(1.7)]{OvMo1996}, when computing $\mathcal O_i$ and $\mathcal O_{i+1}$ (step~1); \emph{(ii)} $\Gamma_i^\dagger$ and $\Gamma_{i-1}^\dagger$ in determining the state sequences (step~5); and \emph{(iii)} $\mathcal Z^\dagger$, with $\mathcal Z\coloneqq [\tilde X_i^\top, U_{i|i}^\top]$, in the least-squares solve for $(A,B,C,D)$ (step~6). For each ridge-regularized pseudoinverse approximation, we solve $(G+\lambda_G I)X=B_G$, where $G\in\mathbb R^{n\times n}$ is the associated Gram matrix, $B_G$ denotes the right-hand side, and $\lambda_G=0.1\operatorname{tr}(G)/n$. Here, $0.1$ is a relative regularization parameter chosen for our implementation.

\end{document}